\documentclass[11pt]{article}

\usepackage[margin=1in]{geometry}
\usepackage{cite}
\usepackage{amsmath,amsfonts,amsthm,amssymb}

\usepackage{enumitem}

\usepackage{algorithmic}
\usepackage{graphicx}
\usepackage{textcomp}
\usepackage{xcolor}

\usepackage{subcaption}
\usepackage{graphicx}
\graphicspath{{figures/}}
\usepackage{color}
\usepackage{url}
\usepackage{soul}

\usepackage{verbatim}
\usepackage{hyperref}
\DeclareMathOperator{\Res}{Res}

\newtheorem{theorem}{Theorem}

\theoremstyle{remark}

\newcommand{\LL}{\mathcal{L}} 

\newcommand{\Nav}{{\mathcal{N}}}

\newcommand{\Div}{\nabla\cdot}

\newcommand{\bu}{{\bf u}}
\newcommand{\bx}{x}
\newcommand{\by}{y}
\newcommand{\bz}{z}
\newcommand{\lam}{\lambda^{*}}

\title{Asymptotic Analysis for the Eigenvalues of Peridynamic Operators}
\author{Bacim Alali,   Nathan Albin, and Thinh Dang\\
\
\footnotesize{Department of Mathematics, Kansas State University, Manhattan, KS}
\thanks{This project is based upon work supported by the National Science Foundation under Grant No. 2108588.}
}

\date{}
\begin{document}

\maketitle
\begin{abstract}
Explicit representations of the eigenvalues of the peridynamic operator have been derived in \cite{alali2023linear}. These representations are given in terms of  generalized hypergeometric functions. Asymptotic analysis of the hypergeometric functions is utilized
to identify the asymptotic behavior of the eigenvalues.
We show that the eigenvalues  are bounded when the kernel is integrable  and diverge when the kernel is
singular. The bounds and decay rates are presented explicitly in terms of
the spatial dimension, the integral kernel and the peridynamic  nonlocality.

\end{abstract}

\section{Introduction}
Peridynamics is a nonlocal formulation of continuum mechanics that is geared towards modeling complex material behavior such as fracture \cite{silling2000reformulation}. The peridynamics model is given by integral equations that do not involve spatial derivatives of the displacement field, which allows for modeling materials' discontinuities \cite{silling_states_2007}.

The peridynamic equation of motion 
for a homogeneous isotropic solid is given by (see for example \cite{silling2000reformulation,silling_states_2007})
\[
\rho(x) \frac{\partial^{2}}{\partial t^2}\bu(x,t)= \LL \bu(\bx,t)+{\bf b}(x,t),
\]
where $\rho$ is the material density, ${\bf b}$ is the applied force, and the operator $\LL$ is given by
\cite{alali2023linear}
\begin{align}
\label{eq:LLdel_homog}
 \nonumber
 \LL \bu(\bx) \nonumber
    &= (n+2)\,\mu\,c\int\limits_{B_\delta(\bx)}
   \frac{(\by-\bx)
\otimes(\by-\bx)}{\|\by-\bx\|^{\beta+2}}\big( \bu(\by)-\bu(\bx)\big)
  \,d\by\\ 
  &\qquad+ \,\frac{(\lam-\mu)c^2}{4}\int\limits_{B_\delta(\bx)}\int\limits_{B_\delta(\by)}
 \frac{\by-\bx}{\|\by-\bx\|^\beta}
\otimes\frac{\bz-\by}{\|\bz-\by\|^\beta}\bu(\bz)
 d\bz d\by.
\end{align}
Here $\bu:\mathbb{R}^n\to \mathbb{R}^n$ denotes the displacement vector-field, $B_\delta(x)$ is the  ball of radius $\delta>0$ centered at $x$, and  the integral kernel exponent satisfies $\beta<n+2$. 

The parameter $\delta$ denotes the nonlocality parameter (also called the {\it peridynamic horizon}), which indicates a nonlocal but finite range interactions, $\mu$ and $\lam$ are the Lam\'{e} parameters of solid mechanics, and
the scaling constant $c$ is 
defined by
\begin{align}
c&:= 
\label{eq:cdel-explicit} 
\frac{2(n+2-\beta)\Gamma\left(\frac{n}{2}+1\right)}    {\pi^{n/2}\delta^{n+2-\beta}}.
\end{align}
With this choice of the scaling constant in \eqref{eq:cdel-explicit}, the peridynamic operator \eqref{eq:LLdel_homog} acting on   sufficiently regular fields  converges to the corresponding  Navier operator  of linear elasticity $\Nav$ for two types of limiting behavior as the nonlocality vanishes either as $\delta\to 0^+$ or as $\beta\to n+2^-$ \cite{alali2023linear}.  Here the Navier operator is given by
\begin{equation}
\label{eq:Nav}
\Nav\bu= (\lam+\mu)\nabla( \Div\bu) + \mu\Delta \bu. 
\end{equation}

For scalar fields $u:\mathbb{R}^n\rightarrow \mathbb{R}$, the analogue to the peridynamic operator $\LL$ in \eqref{eq:LLdel_homog}, is the nonlocal Laplace operator  given by
\cite{alali2019fourier}
\begin{equation}\label{eq:nonlocal_laplacian}
  L u(x) = c\int_{B_\delta(x)}\frac{u(y)-u(x)}{\|y-x\|^\beta}\;dy,
\end{equation}
with $c$ given by \eqref{eq:cdel-explicit}. 
One of the main features of nonlocal Laplace operators of the form \eqref{eq:nonlocal_laplacian} and peridynamic operators of the form \eqref{eq:LLdel_homog} is that the integral kernel is compactly supported and therefore the kernel can be singular, when $n\le \beta <n+2$, or integrable, when $\beta<n$. This  allows for the scalar field $u$ (or the vector-field $\bu$) to have less regularity. For example, for the integrable kernel case when $\beta<n$, the nonlocal Laplacian $L$ is well-defined for $L^1$ fields.  

Nonlocal Laplace operators, such as \eqref{eq:nonlocal_laplacian}, have been studied in the context of nonlocal diffusion, digital image correlation, and nonlocal wave phenomena among other applications, see for example \cite{bobaru2010peridynamicheat,burch2011classical,lehoucq2015novel,seleson2013interface}. Several mathematical and numerical studies have focused on nonlocal Laplace operators including  \cite{radu2017nonlocal,du2012analysis,aksoylu2011variational,alali2020fourier}.


We note that nonlocal scalar phenomena, such as nonlocal diffusion and nonlocal wave propagation, have been also modeled using fractional differential equations, see for example \cite{klafter2012fractional, bucur2016nonlocal, mainardi2010fractional, caffarelli2010drift}.
For the analysis of the fractional Laplacian and related nonlocal equations, see for example \cite{bourgain2001another, brezis2002recognize, caffarelli2007extension, caffarelli2008regularity,agarwal2020special,ruzhansky2017advances}. Works on the eigenvalues of  the fractional Laplacian include \cite{dyda2012fractional,lindgren2014fractional}.  Connections between the fractional Laplacian,  and the nonlocal Laplacian and peridynamic models have been  studied in \cite{du2012analysis, burch2011classical, defterli2015fractional,d2013fractional}. 
However, 
the focus of this work is on the asymptotic behavior of linear peridynamics' eigenvalues. 
\section{Asymptotic behavior of the eigenvalues}
Different representations of the  eigenvalues of the peridynamic operator $\LL$ have been derived in \cite{alali2023linear}. The eigenvalues are characterized through integral representations as well as hypergeometric representations. We focus on the hypergeometric representations as they provide explicit dependence on the spatial dimension $n$ and the nonlocal parameters $\delta$ and $\beta$. 

To simplify the presentation, here and throughout this article, we define
\begin{align*}
    a:=\frac{n+2-\beta}{2},\quad b:=\frac{n+2}{2}\text{ and }z:=z(\nu)=\frac{\delta\|\nu\|}{2}.
\end{align*}
Then, 
for $\nu\in\mathbb{R}^n$, the eigenvalues of $\LL$ are given by
\begin{align}\label{eq:lambda_2}
    \lambda_1(\nu)&=\lambda_{1,1}(\nu)+\lambda_{1,2}(\nu),\\
       \lambda_2(\nu)&=-\mu\|\nu\|^2{}_2F_3\left(1,a;2,b+1,a+1;-z^2\right),
\end{align}
where
\begin{align}
    &\lambda_{1,1}(\nu)=-3\mu\|\nu\|^2 {}_3F_4\left(1,\frac{5}{2},a;2,\frac{3}{2},b+1,a+1;-z^2\right),\label{eq:lambda_{1,1}}\\
    &\lambda_{1,2}(\nu)=
    -\|\nu\|^2(\lambda^*-\mu)\;{}_1F_2\left(a;b,a+1;-z^2\right)^2.\label{eq:lambda_{1,2}}
\end{align}
 
The generalized hypergeometric functions are defined as follows.
\begin{eqnarray*}
    _2F_3(a_1,a_2;b_1,b_2,b_3,z) := \sum_{k=0}^\infty \frac{(a_1)_k(a_2)_k}{(b_1)_k(b_2)_k(b_3)_k}\frac{z^k}{k!},
\end{eqnarray*}
where $(a)_k = a(a+1)(a+2)\cdots(a+k-1)$ is called the rising factorial, also known as the Pochhammer symbol. The hypergeometric functions ${}_1F_2$ and ${}_3F_4$ are defined in a similar way.

The asymptotic behavior of $\lambda_2$ is given by the following result.
\begin{theorem}\label{thm:asymptotics_2}
Let $n\ge 1$, $\delta>0$ and $\beta\le n+2$.  Then, as $\|\nu\|\to\infty$,
\begin{equation*}
\lambda_2(\nu) \sim 
\begin{cases}
-\frac{4\mu ab}{\delta^2(a-1)}-\frac{\Gamma(b+1)\Gamma(a+1)}{\frac{\beta-n}{2}\Gamma(\frac{\beta+2}{2})}\frac{4\mu}{\delta^2}z^{\beta-n}
&\text{if $\beta\ne n$},\\
-\frac{4\mu ab}{\delta^2}\left(
\log z+\gamma-\psi(b)\right)
&\text{if $\beta = n$},
\end{cases}
\end{equation*}
where $\gamma$ is Euler's constant and $\psi$ is the digamma function.
\end{theorem}

\begin{proof}
For large $|z|$,~\cite[Eq.~(16.11.8)]{NIST:DLMF} states that 
\begin{equation*}
_2F_3(1,a;2,b+1,a+1;-z^2) 
\sim a\Gamma(b+1)\left(H_{2,3}(z^2) + E_{2,3}(z^2e^{-i\pi}) + E_{2,3}(z^2e^{i\pi})\right),
\end{equation*}
where $E_{2,3}$ and $H_{2,3}$ are formal series defined in~\cite[Eq.~(16.11.1)]{NIST:DLMF} and~\cite[Eq.~(16.11.2)]{NIST:DLMF} respectively. From those definitions, it yields that
\begin{equation*}
E_{2,3}(z^2e^{\pm i\pi}) 
= (2\pi)^{-1/2}2^{b+2}e^{\pm 2iz}
\sum_{k=0}^{\infty}c_k(\pm 2iz)^{-\frac{2b+5}{2}-k}.
\end{equation*}
Since $2b=n+2$ and $c_0=1$, the $E_{2,3}$ terms decay asymptotically like $|z|^{-\frac{n+7}{2}}$, and do not contribute the the asymptotic behavior described in the theorem.
For the remaining term that relates to the $H_{2,3}$ function, from the remark below~\cite[Eq.~(16.11.5)]{NIST:DLMF}, $H_{2,3}$ can be recognized as the sum of the residues of certain poles of the integrand in~\cite[Eq.~(16.5.1)]{NIST:DLMF}. The integrand for this setting is as following:
\begin{eqnarray*}
f(s) &:=& \frac{\Gamma(1+s)\Gamma(a+s)\Gamma(-s)(z^2)^s}{\Gamma(2+s)\Gamma(b+1+s)\Gamma(a+1+s)}
\\
&=& \frac{\Gamma(-s)z^{2s}}{(s+1)(s+a)\Gamma(b+1+s)}.
\end{eqnarray*}
The two poles of interest are at $s=-1$ and $s=-a$.  The restriction that $\beta\notin\{n+2,n+4,n+6,\ldots\}$ ensures that $-a$ is not a nonnegative integer and, therefore, that $s=-a$ is not a pole of $\Gamma(-s)$.  Thus, there are only two cases to consider, either $\beta\ne n$, which implies that $s=-1$ and $s=-a$ are distinct simple poles of $f$, or $\beta=n$, yielding a double pole at $s=-1$.

\noindent
In the case where $\beta\ne n$, we have
\begin{eqnarray*}
H_{2,3}(z^2) &=& \Res(f,-1) + \Res(f,-a)\\
&=& \frac{z^{-2}}{(a-1)\Gamma(b)} + \frac{\Gamma(a)z^{-2a}}{(1-a)\Gamma(b+1-a)}.
\end{eqnarray*}
When $\beta=n$, the double pole makes the residue more complicated (and gives rise to the logarithmic term):
\begin{align*}
H_{2,3}(z^2)&=\Res(f,-1)\\
&= \left.\frac{d}{ds}\left(\frac{\Gamma(-s)}{\Gamma(b+1+s)}z^{2s}\right)\right|_{s=-1} \\
&= \frac{2\log z-\psi(b)-\psi(1)}{\Gamma(b)}z^{-2}.
\end{align*}
Finally, using~\eqref{eq:lambda_2} and substituting completes the theorem.
\end{proof}

For the asymptotic behavior of $\lambda_1$, we first obtain results for $\lambda_{1,1}$ and $\lambda_{1,2}$.
\begin{theorem}\label{thm:asymptotics_{1,1}}
Let $n\ge 1$, $\delta>0$ and $\beta\le n+2$.  Then, as $\|\nu\|\to\infty$,
\begin{align*}
\lambda_{1,1}(\nu)
&\sim
\begin{cases}
-\frac{4\mu ab}{\delta^2(a-1)}-\frac{n-\beta-1}{n-\beta}\frac{\Gamma(b+1)\Gamma(a+1)}{\Gamma(\frac{\beta+2}{2})}\frac{8\mu}{\delta^2}z^{\beta-n}
&\text{if $\beta\ne n$},\\
-\frac{4\mu ab}{\delta^2}\left(
\log z+\gamma+2-\psi(b)\right)
&\text{if $\beta = n$},
\end{cases}
\end{align*}
where $\gamma$ is Euler's constant and $\psi$ is the digamma function.
\end{theorem}

\begin{proof}
From the setting,~\cite[Eq.~(16.11.8)]{NIST:DLMF} states that for large $|z|$,
\begin{equation*}
{}_3F_4\left(1,\frac{5}{2},a;2,\frac{3}{2},b+1,a+1;-z^2\right)
\sim\frac{2}{3} a\Gamma(b+1)\left(H_{3,4}(z^2) 
+ E_{3,4}(z^2e^{-i\pi}) + E_{3,4}(z^2e^{i\pi})\right),
\end{equation*}
where $E_{3,4}$ and $H_{3,4}$ are formal series defined in~\cite[Eq.~(16.11.1)]{NIST:DLMF} and~\cite[Eq.~(16.11.2)]{NIST:DLMF} respectively. From those definitions, it yields that
\begin{equation*}
E_{3,4}(z^2e^{\pm i\pi}) = (2\pi)^{-1/2}2^{b+1}e^{\pm 2iz}
\sum_{k=0}^{\infty}c_k(\pm 2iz)^{-\frac{2b+3}{2}-k}.
\end{equation*}
Since $2b=n+2$ and $c_0=1$, the $E_{3,4}$ terms decay asymptotically like $|z|^{-\frac{n+5}{2}}$, and do not contribute the the asymptotic behavior described in the theorem.
\noindent
For the remaining term that relates to the $H_{3,4}$ function, from the remark below~\cite[Eq.~(16.11.5)]{NIST:DLMF}, $H_{3,4}$ can be recognized as the sum of the residues of certain poles of the integrand in~\cite[Eq.~(16.5.1)]{NIST:DLMF}. The integrand for this setting is as the following:
\begin{align*}
f(s) 
&:= \frac{\Gamma(1+s)\Gamma(\frac{5}{2}+s)\Gamma(a+s)\Gamma(-s)(z^2)^s}{\Gamma(2+s)\Gamma(\frac{3}{2}+s)\Gamma(b+1+s)\Gamma(a+1+s)}
\\
&=\frac{(\frac{3}{2}+s)\Gamma(-s)z^{2s}}{(s+1)(s+a)\Gamma(b+1+s)}.
\end{align*}
The two poles of interest are at $s=-1$ and $s=-a$.  The restriction that $\beta\notin\{n+2,n+4,n+6,\ldots\}$ ensures that $-a$ is not a nonnegative integer and, therefore, that $s=-a$ is not a pole of $\Gamma(-s)$.  Thus, there are only two cases to consider, either $\beta\ne n$, which implies that $s=-1$ and $s=-a$ are distinct simple poles of $f$, or $\beta=n$, yielding a double pole at $s=-1$.

\noindent
In the case where $\beta\ne n$, we have
\begin{align*}
H_{3,4}(z^2) &= \Res(f,-1) + \Res(f,-a)\\
&= \frac{\frac{1}{2}z^{-2}}{(a-1)\Gamma(b)} 
+ \frac{(\frac{3}{2}-a)\Gamma(a)z^{-2a}}{(1-a)\Gamma(b+1-a)}.
\end{align*}
When $\beta=n$, the double pole makes the residue more complicated (and gives rise to the logarithmic term):
\begin{align*}
H_{3,4}(z^2)&=\Res(f,-1)\\
&= \left.\frac{d}{ds}\left(\frac{(\frac{3}{2}+s)\Gamma(-s)}{\Gamma(b+1+s)}z^{2s}\right)\right|_{s=-1}\\
&=\left[\frac{\Gamma(-s)}{\Gamma(b+1+s)}z^{2s}\right.\\
&\quad\left.\left. +\left(\frac{3}{2}+s\right)\frac{d}{ds}\left(\frac{\Gamma(-s)}{\Gamma(b+1+s)}z^{2s}\right)\right]\right|_{s=-1}\\
&= \frac{2\log z-\psi(b)-\psi(1)+2}{2\Gamma(b)}z^{-2}.
\end{align*}
Finally, using~\eqref{eq:lambda_{1,1}} and substituting completes the theorem.
\end{proof}

\begin{theorem}\label{thm:asymptotics_{1,2}}
Let $n\ge 1$, $\delta>0$ and $\beta\le n+2$.  Then, as $\|\nu\|\to\infty$,
\begin{equation*}
\lambda_{1,2}(\nu) \sim 
-(\lambda ^*-\mu)\left[\frac{\Gamma(b)\Gamma(a+1)}{\Gamma(\frac{\beta}{2})}\right]^2\frac{4}{\delta^2}z^{2(\beta-(n+1))}.
\end{equation*}
\end{theorem}

\begin{proof}
From ~\cite[Eq.~(16.11.8)]{NIST:DLMF}, it implies that for large $|z|$,
\begin{equation*}
_1F_2(a;b,a+1;-z^2) 
\sim a\Gamma(b)\left(H_{1,2}(z^2) + E_{1,2}(z^2e^{-i\pi}) + E_{1,2}(z^2e^{i\pi})\right),
\end{equation*}
where $E_{1,2}$ and $H_{1,2}$ are formal series defined in~\cite[Eq.~(16.11.1)]{NIST:DLMF} and~\cite[Eq.~(16.11.2)]{NIST:DLMF} respectively. From those definitions, it yields that
\begin{equation*}
E_{1,2}(z^2e^{\pm i\pi}) 
= (2\pi)^{-1/2}2^{b+1}e^{\pm 2iz}
\sum_{k=0}^{\infty}c_k(\pm 2iz)^{-\frac{2b+1}{2}-k}.
\end{equation*}
Since $2b=n+2$ and $c_0=1$, the $E_{1,2}$ terms decay asymptotically like $|z|^{-\frac{n+3}{2}}$, and do not contribute the the asymptotic behavior described in the theorem.
\noindent
For the remaining term that relates to the $H_{1,2}$ function, from the remark below~\cite[Eq.~(16.11.5)]{NIST:DLMF}, $H_{1,2}$ can be recognized as the sum of the residues of certain poles of the integrand in~\cite[Eq.~(16.5.1)]{NIST:DLMF}. The integrand for this setting is as following:
\begin{eqnarray*}
f(s) &:=& \frac{\Gamma(a+s)}{\Gamma(b+s)\Gamma(a+1+s)}
\Gamma(-s)(z^2)^s\\
&=& \frac{\Gamma(-s)}{(s+a)\Gamma(b+s)}z^{2s}.
\end{eqnarray*}
The only pole of interest is at $s=-a$.  The restriction that $\beta\notin\{n+2,n+4,n+6,\ldots\}$ ensures that $-a$ is not a non-negative integer and, therefore, that $s=-a$ is not a pole of $\Gamma(-s)$.  Thus
\begin{equation*}
H_{1,2}(z^2) = \Res(f,-a)
= \frac{\Gamma(a)}{\Gamma(b-a)}z^{-2a}.
\end{equation*}
Finally, using~\eqref{eq:lambda_{1,2}} and substituting completes the theorem.
\end{proof}

The asymptotic behavior of $\lambda_1(\nu)$ can be summarized in the following theorem.
\begin{theorem}\label{thm:asymptotics_1}
Let $n\ge 1$, $\delta>0$ and $\beta\le n+2$.  Then, as $\|\nu\|\to\infty$,
\[
\lambda_{1}(\nu) \sim 
-(\lambda ^*-\mu)\!\left[\frac{\Gamma(b)\Gamma(a+1)}{\Gamma(\frac{\beta}{2})}\right]^2\frac{4}{\delta^2}z^{2(\beta-(n+1))}
+
\begin{cases}
-\frac{4\mu ab}{\delta^2(a-1)}-\frac{n-\beta-1}{n-\beta}\frac{\Gamma(b+1)\Gamma(a+1)}{\Gamma(\frac{\beta+2}{2})}\frac{8\mu}{\delta^2}z^{\beta-n}
\!\!\!\!&\text{if $\beta\ne n$},\\
-\frac{4\mu ab}{\delta^2}\left(
\log z+\gamma+2-\psi(b)\right)
&\text{if $\beta = n$},
\end{cases}
\]
where $\gamma$ is Euler's constant and $\psi$ is the digamma function.
\end{theorem}

\begin{figure*}
\centering
\includegraphics[width=\textwidth]{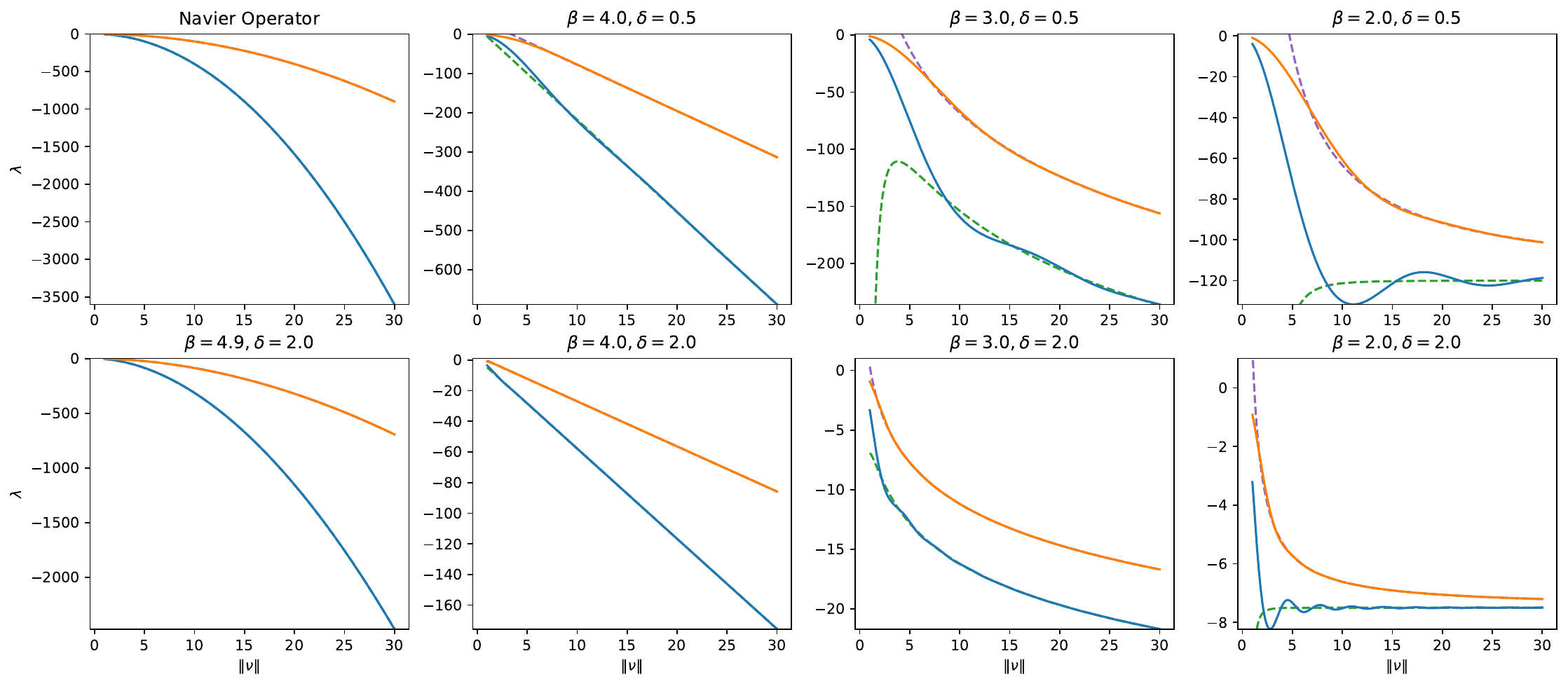}
\caption{Comparison between the eigenvalues, $\lambda_1(\nu)$ and $\lambda_2(\nu)$, and their asymptotic approximations given in Theorems~\ref{thm:asymptotics_1} and~\ref{thm:asymptotics_2} for the 3D case. Here $\|\nu\|$ (horizontal axis) is sampled at $1000$ equispaced points in the interval $[0,30]$ and $\delta$ and $\beta$ are as given in the titles.  The shear modulus and the second Lam\'{e} parameter are given by $\mu=1$ and $\lam=2$. For each plot, the upper solid curve shows $\lambda_2(\nu)$, the lower solid curve shows $\lambda_1(\nu)$, and the dashed curves show the corresponding asymptotic approximations.}
\label{fig:asymptotics3D}
\end{figure*}

\begin{figure*}
\centering
\includegraphics[width=\textwidth]{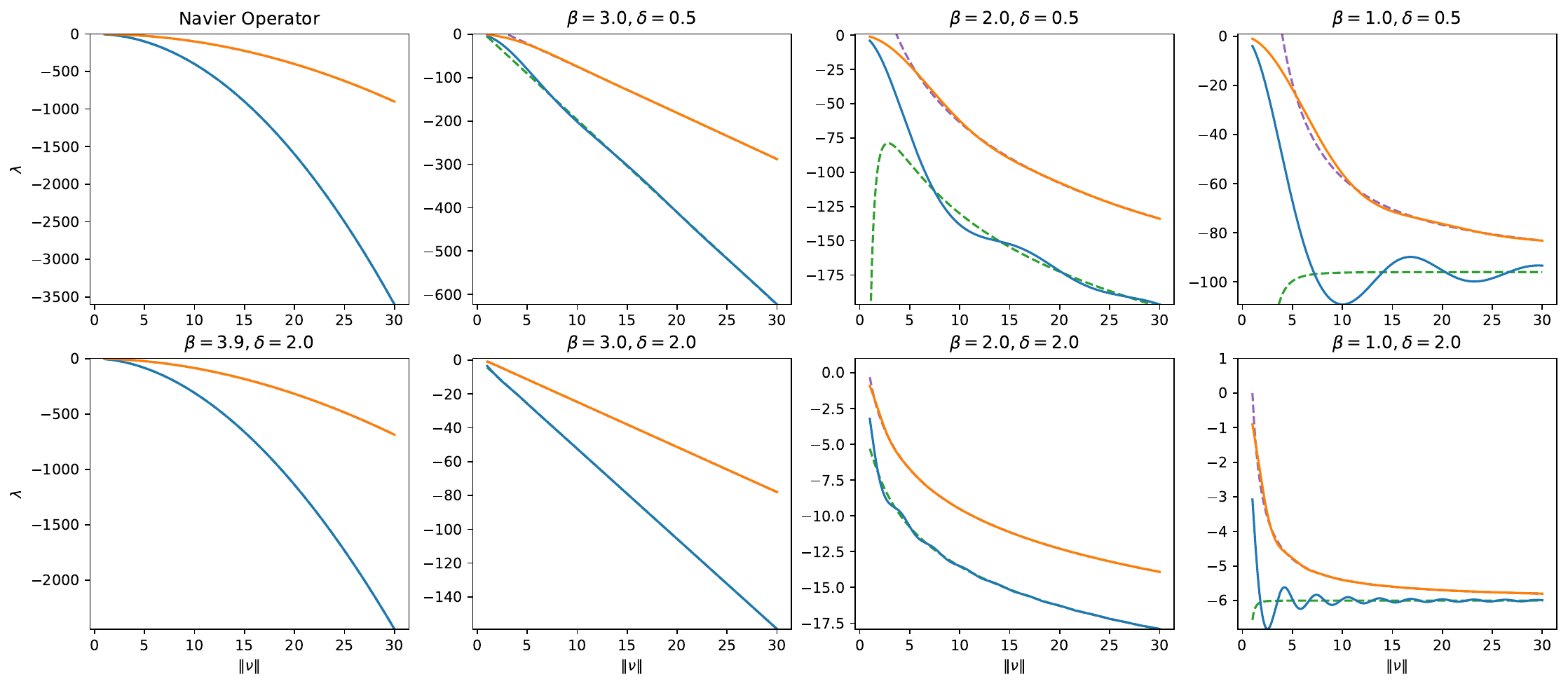}
\caption{Comparison between the eigenvalues, $\lambda_1(\nu)$ and $\lambda_2(\nu)$, and their asymptotic approximations given in Theorems~\ref{thm:asymptotics_1} and~\ref{thm:asymptotics_2} for the 2D case. Here $\|\nu\|$ (horizontal axis) is sampled at $1000$ equispaced points in the interval $[0,30]$ and $\delta$ and $\beta$ are as given in the titles.  The shear modulus and the second Lam\'{e} parameter are given by $\mu=1$ and $\lam=2$. For each plot, the upper solid curve shows $\lambda_2(\nu)$, the lower solid curve shows $\lambda_1(\nu)$, and the dashed curves show the corresponding asymptotic approximations.}
\label{fig:asymptotics2D}
\end{figure*}

\section{Discussion}

Figures~\ref{fig:asymptotics3D} and \ref{fig:asymptotics2D} show the eigenvalues of the three-dimensional and the two-dimensional operator, $\mathcal{L}$, along with the corresponding asymptotic approximations for several choices of peridynamic parameters, $\beta$ and $\delta$. Note that the overall shape of the eigenvalue curves are determined by the parameter, $\beta$, while $\delta$ affects the scaling of the graphs. When $\beta<n$, all eigenvalues of $\mathcal{L}$ are bounded. For $\beta \ge n$, all eigenvalues are unbounded; they diverge toward $-\infty$ at a rate given by the asymptotic formulas. When $\beta=n$, this rate is logarithmic, otherwise the eigenvalues diverge toward $-\infty$ at a rate comparable to $\|\nu\|^{\beta-n}$. For both the three-dimensional and the second dimensional cases, the plots show asymptotic linear growth when $\beta=n+1$ and approach quadratic growth as $\beta$ approaches $n+2$. From the plots, one can observe that the eigenvalues of $\mathcal{L}$ converge pointwise to those of the Navier operator~\eqref{eq:Nav} as $\beta\to n+2$, a fact that can be verified by substituting $\beta=n+2$ into the eigenvalue formulas~\eqref{eq:lambda_{1,1}},~\eqref{eq:lambda_{1,2}} and~\eqref{eq:lambda_2}. Similar convergence of the eigenvalues to those of $\mathcal{N}$ can be observed in taking $\delta\to 0$~\cite{alali2023linear}.

The asymptotic growth rates (or boundedness) of the eigenvalues of $\mathcal{L}$, as given in Theorem~\ref{thm:asymptotics_2} and Theorem~\ref{thm:asymptotics_1}, is a valuable tool in the study of peridynamic equations. Similar approximations have been instrumental in establishing the regularity of solutions for nonlocal scalar problems, such as the nonlocal analogs of the Poisson and heat equations~\cite{alali2019fourier,mustapha2023regularity}. Moreover, knowledge of the asymptotic behavior of the eigenvalues has proven useful in establishing convergence of solutions, e.g., the convergence of solutions to the nonlocal Poisson and heat equations to the solutions of the corresponding classical equations
as $\beta\to n+2$ or $\delta\to 0$~\cite{alali2019fourier,mustapha2023regularity}. Pointwise convergence of the multipliers is not sufficient to prove convergence of solutions.
It is expected that the asymptotic results given in Theorems \ref{thm:asymptotics_2}  and \ref{thm:asymptotics_1} will lead to analogous results for peridynamic vector equations.

\bibliographystyle{acm}
\bibliography{refs}

\end{document}